\RequirePackage{amsmath} %
\documentclass[a4paper]{article}
\usepackage[utf8]{inputenc}
\usepackage{amssymb}
\usepackage{pifont}%
\newcommand{\cmark}{\ding{51}}%
\newcommand{\xmark}{\ding{55}}%
\usepackage{todonotes}
\usepackage{algorithm}  
\usepackage[noend]{algorithmic}
\usepackage{times}
\usepackage{url}
\usepackage{xcolor,svgcolor}
\usepackage[all]{xy} 
\usepackage{rotating}
\usepackage{paralist}
\usepackage{xspace}
\usepackage{epsfig}
\usepackage{subfigure}
\usepackage{calc}
\usepackage{amstext}
\usepackage{amsmath}
\usepackage{amsthm}
\usepackage{multicol,multirow}
\usepackage{pslatex}
\usepackage{hyperref}
\usepackage{authblk}
\usepackage{fullpage}

 \hypersetup{
  pdftitle={Private Multi-party Matrix Multiplication and Trust Computations},
  pdfauthor={Dumas, Lafourcade, Orfila, Puys},
  breaklinks=true,
  colorlinks=true,
  linkcolor=blue,
  citecolor=darkred,
  urlcolor=darkblue,
}
\makeatletter
\newcommand{\IFTHENELSE}[3]{\ALC@it\algorithmicif\ #1\
  \algorithmicthen\ #2\
  \algorithmicelse\ #3\
  \ifthenelse{\boolean{ALC@noend}}{}{\algorithmicendif\ } }
\newcommand{\FOREND}[3][default]{\ALC@it\algorithmicfor\ #2\ \algorithmicdo%
  \ALC@com{#1}\ #3}
\makeatother

\newcommand{\E}{\ensuremath{\mathbb{E}}}
\newcommand{\F}{\ensuremath{\mathbb{F}}}
\newcommand{\D}{\ensuremath{\mathbb{D}}}
\newcommand{\T}{{\ensuremath{\D^2}}}
\newcommand{\couple}[2]{{\ensuremath{\langle{}#1,#2\rangle}}}
\newcommand{\paragg}{{\ensuremath{\maltese}}\xspace}
\newcommand{\seqagg}{{\ensuremath{\bigstar}}\xspace}

\newcommand{\HAdd}[2]{{\ensuremath{\text{Add}(#1;#2)}}}
\newcommand{\HMul}[2]{{\ensuremath{\text{Mul}(#1;#2)}}}

\newcommand{\N}{\ensuremath{\mathbb{N}}}
\newcommand{\Z}{\ensuremath{\mathbb{Z}}} 
\newcommand{\Zpz}[1]{\ensuremath{\Z{/#1\Z}}}
\newcommand{\ZN}{\Zpz{N}}

\newcommand{\Zm}{\Zpz{m}}

\newcommand{\bigO}[1]{{\ensuremath{\mathcal{O}\left(#1\right)}}}
\newtheorem{theorem}[algorithm]{Theorem}
\newtheorem{corollary}[algorithm]{Corollary}
\newtheorem{lemma}[algorithm]{Lemma} 
\newtheorem{definition}[algorithm]{Definition}

\newtheorem{proposition}[algorithm]{Proposition}

\newcommand{\proverif}{ProVerif}

\newcommand{\overlongrightarrow}[2]
{\vbox{\offinterlineskip
      \halign{##\cr
              \hfill #1 \hfill\cr
              \hbox to #2{\relax\rightarrowfill}\cr}}\ignorespaces
}

\newcommand{\overlongleftarrow}[2]
{\vbox{\offinterlineskip
      \halign{##\cr
              \hfill #1 \hfill\cr
              \hbox to #2{\relax\leftarrowfill}\cr}}\ignorespaces
}

\newcommand{\mytitle}{Private Multi-party Matrix Multiplication and
 Trust Computations\thanks{This work was partially supported by ``Digital
    trust'' Chair from the University of Auvergne Foundation, by the HPAC project (ANR~11~BS02~013), the ARAMIS project (PIA P3342-146798) and the LabEx PERSYVAL-Lab (ANR-11-LABX-0025).}}

\title{\mytitle}

\author[1]{Jean-Guillaume~Dumas}
\author[2]{Pascal~Lafourcade}
\author[1]{Jean-Baptiste~Orfila}
\author[3]{Maxime~Puys}

\affil[1]{Universit\'e Grenoble Alpes, CNRS, LJK, 700 av. centrale,
  IMAG/CS-40700, 38058 Grenoble cedex 9, France.
\href{mailto:Jean-Guillaume.Dumas@imag.fr,Jean-Baptiste.Orfila,Maxime.Puys@imag.fr}{\{Jean-Guillaume.Dumas,Jean-Baptiste.Orfila\}@imag.fr}
}
\affil[2]{Universit\'e Clermont Auvergne, LIMOS, Campus Universitaire
  des C\'ezeaux, BP 86, 63172 Aubi\`ere Cedex, France. 
\href{mailto:Pascal.Lafourcade@udamail}{Pascal.Lafourcade@udamail}}
\affil[3]{Universit\'e Grenoble Alpes, CNRS, Verimag, 700
  av. centrale, IMAG - CS 40700, 38058 Grenoble cedex 9, France.
  \href{mailto:Maxime.Puys@imag.fr}{Maxime.Puys@imag.fr}}

\begin{document}

\onecolumn\maketitle%

\abstract{ 
  This paper deals with distributed matrix
  multiplication. Each player owns only one row of both
  matrices and wishes to learn about one distinct row of
  the product matrix, without revealing its input to the other
  players. 
  We first improve on a weighted average protocol, in order to 
  securely compute a dot-product with a quadratic volume of
  communications and linear number of rounds. 
  We also propose a protocol with five communication rounds, using
  a Paillier-like underlying homomorphic public key
  cryptosystem, which is secure in the semi-honest model or secure with high
  probability in the malicious adversary model.  Using \proverif{}, a
  cryptographic protocol verification tool, we are able to check the
  security of the protocol and provide a countermeasure for each
  attack found by the tool. We also give a randomization method to
  avoid collusion attacks.  As an application, we show
  that this protocol enables a distributed and secure evaluation of
  trust relationships in a network, for a large class of trust
  evaluation schemes.}

\section{Introduction}
Secure multiparty computations (MPC), introduced by
Yao~\cite{Yao:1982} with the millionaires' problem, has been
intensively studied during the past thirty years. The idea of MPC is
to allow $n$ players to jointly compute a function $f$ using their
private inputs without revealing them. In the end, they only know the
result of the computation and no more information. Depending on
possible corruptions of players, one may prove 
that a protocol may resist against a collusion of many players, or
that it is secure even if attackers try to maliciously modify their
inputs. Mostly any function can be securely
computed~\cite{Ben-Or:1988} and many
tools exist to realize MPC protocols.  They comprise for instance the
use of a Trusted Third Party~\cite{Du:2002}, the use of Shamir's
secret sharing scheme~\cite{Shamir:1979}, or more recently the use of
homomorphic encryption~\cite{Goethals:05}. It is also possible to mix
these techniques~\cite{Damgard:2012}.

Our goal is to apply MPC to the a distributed evaluation of trust, as defined
in~\cite{Josang:2007:PLUU,jgd:2012:pkitrust}. There, confidence is a
combination of degrees of trust, distrust and 
uncertainty between players. Aggregation of trusts between players on
a network is done by a matrix product defined on two monoids (one for
the addition of trust, the other one for multiplication, or
transitivity): each player knows one row of the matrix,
its partial trust on its neighbors, and the network as a whole has to
compute a distributed matrix squaring.  Considering that the trust of
each player for his colleagues is private, at the end of the
computation, nothing but one row of the global trust has to be learned
by each player (\emph{i.e.,} nothing about private inputs
should be revealed to others).  Thus, an MPC protocol to resolve
this problem should combine privacy (nothing is learned but the
output), safety (computation of the function does not reveal anything
about inputs) and efficiency~\cite{Lindell:09}.
First, we need to define a MPC protocol which allows us to efficiently
compute a distributed matrix product with this division of data
between players. The problem is reduced to the computation of a dot
product between vectors $U$ and $V$ such that one player knows $U$ and
$V$ is divided between all players.\\
\noindent {\bf Related Work.} Dot product in the MPC model has been
widely studied~\cite{Du:2001:PCS,Amirbekyan:2007:dotprod,Wang:2009}. 
However, in these papers, assumptions made on data partitions are
different: there, each player owns a complete vector, and the
dot product is computed between two players where; in our setting, trust
evaluation should be done among peers, like certification
authorities. For instance, using a trusted third party or permuting
the coefficients is unrealistic.  
Now, computing a dot product with $n$ players is actually
close to 
the MPWP protocol of~\cite{Dolev:2010:CMT}, computing a mean in a
distributed manner: computing dot products is actually similar to
computing a weighted average where the weights are in the known row,
and the values to be averaged are privately distributed. In MPWP the
total volume of communication for a dot product is 
 \bigO{n^3} with \bigO{n} communication rounds. 
Other generic MPC protocols exist, also evaluating circuits, they
however also require \bigO{n^3} computations and/or communications per
dot-product~\cite{Bendlin:2011:eurocrypt,Damgard:2012}.
\par\noindent {\bf Contributions.} We provide the following results:
\begin{compactitem}
\item A protocol \textit{P-MPWP}, improving on \textit{MPWP}, which
  reduces both the computational cost, by allowing the use of Paillier's
  cryptosystem, and the communication cost, from $\bigO{n^3}$ to $\bigO{n^2}$.
\item An $\bigO{n}$ time and communications protocol \emph{Distributed and Secure
  Dot-Product} ($DSDP_i$) (for $i$ participants) which
allows us to securely compute a dot product $UV$, against a semi-honest
adversary, where one player owns a vector $U$ and where each player
knows one coefficient of $V$.
\item A parallel variant that performs the dot-product computation in
  parallel among the players, limits the total number of rounds. This
  is extended to a \emph{Parallel Distributed and Secure
    Matrix-Multiplication} ($PDSMM_i$) family of protocols.
\item A security analysis of the $DSDP$ protocol using a cryptographic
  protocol verification tool, here 
\proverif{}~\cite{blanchet01,ProVerif}. This tool allows us to
define countermeasures for each found attack:
adapted proofs of knowledge in order to preserve privacy 
and
a \emph{random ring order}, where private inputs are protected
  as in a wiretap code~\cite{Ozarov:1984:wiretapII} and where the players 
take random order in the protocol  to preserve privacy with high
probability, even against a coalition of malicious insiders.
\item Finally, we show how to use these protocols for the computation of
  trust aggregation, where classic addition and multiplication are
  replaced by more generic operations, defined on monoids.
\end{compactitem}

In Section~\ref{sec:background}, we thus first recall some multi-party
computation notions. 
We then introduce in Section~\ref{sec:monoid} the trust model based on
monoids. In Section~\ref{sec:p-mpwp}, we present our quadratic variant
of MPWP and a linear-time protocol in Section~\ref{sec:protocol}. We
then give the associated security proofs and countermeasures in
Section~\ref{sec:security} and present parallelized version in
Section~\ref{sec:matmul}. 
Finally, in Section~\ref{sec:mpcmonoid}, we show how our protocols can
be adapted to perform a private multi-party trust computation in a
network.

\section{Background and Definitions}\label{sec:background}
We use a public-key homomorphic encryption scheme where both addition
and multiplication are considered.  There exist many homomorphic
cryptosystems, see for instance~\cite[\S~3]{Mohassel:2011:seclinalg}
and references therein.  We need the following properties on the
encryption function $E$ (according to the context, we use $E_{PubB}$,
or $E_1$ or just $E$ to denote the encryption function, similarly for
the signature function, $D_1$ or $D_{privB}$): computing several
modular additions, denoted by $\HAdd{c_1}{c_2}$, on ciphered messages and one
modular multiplication, denoted by $\HMul{c}{m}$, between a ciphered message and
a cleartext.  That is, $\forall m_1, m_2 \in \Zm$:
$ \HAdd{E(m_1)}{E(m_2)} = E(m_1 + m_2 \mod m)$ and 
$ \HMul{E(m_1)}{m_2} =  E(m_1m_2\mod m)$.
For instance, Paillier's or Benaloh's
cryptosystems~\cite{Paillier:1999:homomorph,Benaloh94,Fousse:2011:benaloh}
can satisfy these requirements, \emph{via} multiplication in the
ground ring for addition of enciphered messages
($\HAdd{E(m_1)}{E(m_2)}=E(m_1)E(m_2)\mod m$), and \emph{via} exponentiation
for ciphered multiplication ($\HMul{E(m_1)}{m_2}=E(m_1)^{m_2}\mod m$), we
obtain the following homomorphic properties:
\begin{align}
\label{homo:add} E(m_1)E(m_2) &= E(m_1 + m_2 \mod m)\\
\label{homo:mul} E(m_1)^{m_2} &=  E(m_1m_2\mod m)
\end{align}

Since we consider the semantic security of the cryptosystem, we assume
that adversaries are probabilistic polynomial time machines. In MPC,
most represented intruders are the following ones:
\begin{compactitem}
\item \emph{Semi-honest (honest-but-curious) adversaries}: a corrupted player
   follows the protocol specifications, but also tries to gather as many 
  information as possible in order to deduce some private inputs.
\item \emph{Malicious adversaries}: a corrupted player that controls
  the network and stops, forges or listens to messages in
  order to gain information.
\end{compactitem}

\section{\uppercase{Monoids of Trust}}
\label{sec:monoid}
There are several schemes for evaluating the transitive trust in a
network.  Some use a single value representing
the probability that the expected action will happen; the
complementary probability being an uncertainty on the trust.  Others
include the \emph{distrust} degree indicating the probability that the
opposite of the expected action will happen \cite{Guha:2004:PTD}.
More complete schemes can be introduced to evaluate trust: J{\o}sang
introduces the \emph{Subjective Logic} notion which expresses beliefs
about the truth of propositions with degrees of "uncertainty"
in~\cite{Josang:2007:PLUU}. Then the authors of~\cite{Huang:2010:FCT}
applied the associated calculus of trust to public key
infrastructures.  There, trust is represented by a triplet, (trust,
distrust, uncertainty) for the proportion of experiences proved, or
believed, positive; the proportion of experiences {\em proved
  negative}; and the proportion of experiences with unknown character.
As $uncertainty=1-trust-distrust$, it is sufficient to express trust
with two values as \couple{trust}{distrust}.  In
\emph{e.g.}~\cite{Foley:2010:stm}
algorithms are proposed to quantify the trust relationship between two
entities in a network, using transitivity and reachability.
For instance, in~\cite{jgd:2012:pkitrust} the authors use an adapted
power of the adjacency matrix to evaluate the trust using all existing
(finite) trust paths between entities.
We show in the following of this section, that powers of this
adjacency matrix can be  evaluated privately in a distributed manner,
provided than one disposes of an homomorphic cryptosystem satisfying 
the homomorphic Properties~(\ref{homo:add}) and~(\ref{homo:mul}).

\subsection{Aggregation of Trust}
Consider Alice trusting Bob with a certain trust degree, and Bob
trusting Charlie with a certain trust degree. 
The sequential aggregation of trust formalizes a kind of transitivity 
 to help Alice to make a decision about Charlie, that
is based on Bob's opinion.  
In the following, we first consider that the trust values are given as
a pair $\couple{a}{b}\in\T$, for $\D$ a principal ideal ring: %
for three players $P_1$, $P_2$ and $P_3$, where $P_1$ trusts $P_2$
with trust value $\couple{a}{b}\in\T$ and $P_2$ trusts $P_3$ with trust value
$\couple{c}{d}\in\T$ the associated {\em sequential aggregation of trust} is a 
function $\seqagg:\T\times\T\rightarrow\T$, that computes the trust
value over the trust path $P_1\stackrel{\couple{a}{b}}{\rightarrow}
P_2 \stackrel{\couple{c}{d}}{\rightarrow} P_3$ as
$\couple{a}{b}\seqagg\couple{c}{d} =\couple{ac+bd}{ad+bc}$.
Similarly, from Alice to Charlie, there might be several ways to perform a
sequential aggregation (several paths with existing trust
values). Therefore it is also possible to aggregate these parallel
paths with the same measure, in the following way:
for two disjoint paths $P_1\stackrel{\couple{a}{b}}{\rightarrow} P_3$ and $P_1
\stackrel{\couple{c}{d}}{\rightarrow} P_3$, the associated {\em parallel
  aggregation of trust} is a function $\paragg:\T\times\T \rightarrow\T$, that computes 
the resulting trust value as: $\couple{a}{b}\paragg\couple{c}{d}=\couple{a+c-ac}{bd}$.
We prove  the following Lemma.
\begin{lemma}
\couple{a}{b} is
invertible for \paragg~if and only if ($b$ is invertible in \D) and (%
$a=0$ or $a-1$ is invertible).
\end{lemma}
\begin{proof}
As $\langle{}a+0-a.0$, $b.1\rangle$=\couple{a}{b}, 
$\couple{0}{1}$ is neutral for \paragg.
Then, for $b$ invertible, if $a=0$, then
\couple{0}{b^{-1}} is an inverse for \couple{0}{b}.
Otherwise, for $a-1$ invertible,
$\couple{a(a-1)^{-1}}{b^{-1}}\paragg\couple{a}{b}=$ $\couple{a}{b}\paragg\couple{a(a-1)^{-1}}{b^{-1}}$ $=\couple{a+a(a-1)^{-1}-a^2(a-1)^{-1}}{bb^{-1}}=$ $\couple{0}{1}$.
Similarly, if $\couple{a}{b}\paragg\couple{c}{d}=\couple{0}{1}$,
then $bd=1$ and $b$ is invertible.
 Then also $(a-1)c=a$. Finally if
$a\neq 0$ and $a-1$ is a zero divisor, there exists $\lambda\neq 0$
such that
$\lambda(a-1)=0$, thus $\lambda(a-1)c=0=\lambda{}a$, but then
$\lambda(a-1)-\lambda{}a=-\lambda=0$. 
As this is contradictory, the
only possibilities are $a=0$ or $a-1$ invertible.
\end{proof}

\subsection{Multi-party Private Aggregation}
For $E$ an encryption function, we define the natural morphism
on pairs, so that it can be applied to trust values:  
$%
E(\couple{a}{b})=\couple{E(a)}{E(b)}%
.$%
We can thus extend homomorphic
properties to pairs so that the parallel and sequential aggregation can
then be computed homomorphically, provided that one entry is in clear.
\begin{lemma} With an encryption function $E$, satisfying the
  homomorphic Properties~(\ref{homo:add}) and~(\ref{homo:mul}), we have:
\begin{align*}
\HMul{E\left(\couple{a}{b}\right)}{\couple{c}{d}}&=E\left(\couple{a}{b}\seqagg\couple{c}{d}\right)\\&=\couple{E(a)^cE(b)^d}{E(a)^dE(b)^c}\\
\HAdd{E\left(\couple{a}{b}\right)}{\couple{c}{d}}&=E\left(\couple{a}{b}\paragg\couple{c}{d}\right)\\&=\couple{E(a)E(c)E(a)^{-c}}{E(b)^d}
\end{align*}
Moreover, those two functions can be computed on an enciphered
\couple{a}{b}, provided that \couple{c}{d} is in clear.
\end{lemma}
\begin{proof}
From the homomorphic properties of the encryption functions, we have:
$E(a)^cE(b)^d=E(ac+bd)$,
$E(a)^dE(b)^c=E(ad+bc)$,
$E(a)E(c)E(a)^{-c}=E(a+c+a(-c))$ and
$E(b)^d=E(bd)$.
For the computation, both right hand sides depend only on
ciphered values $E(a)$, $E(b)$, and on clear values $c$ and $d$
($E(c)$ can be computed with the public key, from~$c$). 
\end{proof}

This shows, that in order to compute the aggregation of trust
privately, the first step is to be able to compute dot-products
privately.

\section{\uppercase{From MPWP to P-MPWP}}
\label{sec:p-mpwp}
\subsection{MPWP description}
The \textit{MPWP} protocol~\cite{Dolev:2010:CMT} is used to securely
compute private trust values in an additive reputation system between
$n$ players. Each player $P_i$ (excepted $P_1$, assumed to be the
master player) has a private entry $v_i$, and $P_1$ private entries
are weights $u_i$ associated to others players.  The goal is to
compute a weighted average trust, \emph{i.e.,}
$\sum_{i=2}^nu_i*v_i$. The idea of \textit{MPWP} is the following: the
first player creates a vector $TV$ containing her private entries
ciphered with her own public key using Benaloh's cryptosystem,
\emph{i.e.,} $TV=[E_1(w_2),\ldots,E_1(w_n)]$. Then, $P_1$ also sends a
$(n-1) \times (n-1)$ matrix $M$, with all coefficients initialized to~$1$ and a variable $A=1$.  Once $(M,TV,A)$ received, each player
computes: $ A = A*E_1(u_i)^{v_i}*E_1(z_i)$, where $z_i$ is a random
value generated by $P_i$.  At the end, the first player gets $D_1(A)=
\sum_{i=2}^n u_iv_i + z_i$. Then, the idea is to cut the $z_i$ values
in $n-1$ positive shares such that $z_i = \sum_{j=2}^{n}z_{i,j}$.
Next, each $z_{i,j}$ is ciphered with the public key
of $P_j$, the result is stored into the $i^{th}$ column of $M$, and
$M$ is forwarded to the next player.  In a second phase, players
securely remove the added random values to $A$, from 
$M=(m_{i,j})=(E_j(z_{i,j}) )$: each player $P_j$,
except $P_1$, computes her $PSS_j=\sum_{i=2}^n D_j(m_{i,j})=
\sum_{i=2}^nz_{i,j}$ by deciphering all values contained in the
$j^{th}$ row of $M$; then they send $\gamma_j=E_1(PSS_j)$ to $P_1$,
their $PSS_i$ ciphered with the public key of $P_1$.  At the end,
$P_1$ retrieves the result by computing $Trust = D_1(A) -
\sum_{j=2}^nD_1(\gamma_j)=D_1(A) - \sum_{j=2}^nPSS_j=D_1(A) -
\sum_{j=2}^n\sum_{i=2}^nz_{i,j}=D_1(A) -\sum_{i=2}^n
z_i=\sum_{i=2}^nu_iv_i$.

 \subsection{P-MPWP: A lighter MPWP}
\textit{P-MPWP} is a variant of \textit{MPWP} with two main
differences: first Paillier's
cryptosystem is used instead of Benaloh's, and, second, the overall
communications cost is reduced from $\bigO{n^3}$ to $\bigO{n^2}$ by
sending parts of the matrix only. 
All steps of \textit{P-MPWP} but
those clearly 
identified in the following are common with \textit{MPWP}, including
the players' global settings. 
Since \textit{P-MPWP} is using a cryptosystem where players can have
different modulus, some requirements must be verified in the players'
settings. First of all, a bound $B$ needs to be fixed for the
vectors' private coefficients:
\begin{equation}\label{hyp:bound}
\forall i, 0\leq u_i\leq B, 0\leq v_i\leq B
\end{equation}
With Benaloh, the common modulus $M$ must be greater than the dot product,
thus at most:
\begin{equation}\label{hyp:benaloh}  
  (n-1)B^2<M.
\end{equation} 
Differently, with Paillier, each player $P_i$ has
a different modulus $N_i$.
Then, by following the $MPWP$ protocol steps, at the end of the first
round, $P_1$ obtains $A=\prod_{i=2}^nE_1(u_i)^{v_i}*E_1(z_i)$. In order
to correctly decipher this coefficient, if the players' values, as
well as their random values $z_i$, satisfy the
bound~(\ref{hyp:bound}), her
modulo $N_1$ must be greater than $(n-1)(B^2+B)$. For others players,
there is only one deciphering step, at 
the second round. They received $(n-1)$ shares all
bounded by $B$. Hence, their modulus $N_i$ need only be greater than
$(n-1)B$. These modulus requirements are summarized in the following
 lemma:
\begin{lemma}
Let $n>3$ be the number of players. Under the bound~(\ref{hyp:bound}),
if $\forall i, 0\leq z_i\leq B$ and if also the modulus satisfy
$(n-1)(B^2+B)<N_1~\text{and}~(n-1)B<N_i,~\forall{}i=2, \ldots,n$, then
at the end of \textit{P-MPWP}, $P_1$ obtains $S_n
=\sum_{i=2}^nu_i*v_i$.
\end{lemma}

Now, the reduction of the communications cost in \textit{P-MPWP}, is
made by removing the exchange of the full $M$ matrix between
players. At the $z_{i,j}$ shares computation, each $P_i$ directly
sends the $j^{th}$ coefficient to  
the $j^{th}$ player instead of storing results in $T$. In the end,
each player $P_i$ receives $(n-1)$ values ciphered with his public key, and
he can compute the $PSS_i$ by deciphering and adding each received
values, exactly as in $MPWP$.
Thus, each player sends only $\bigO{n}$ values,
instead of $\bigO{n^2}$. All remaining steps can be executed as in
\textit{MPWP}. 

Both Paillier's and Benaloh's cryptosystems provides semantic
security, thus the security of \textit{P-MPWP} is not altered. 
Moreover,
since a common bound is fixed \textit{a priori} on private inputs, 
\textit{P-MPWP} security can be reduced to the one in \textit{MPWP}
with the common modulo $M$ between all players~\cite{Michalas2012}. 
Finally, since all exploitable (\emph{i.e.,} clear or ciphered with the
dedicated key) 
information exchanged represents a 
subset of the 
\textit{MPWP} players' knowledge, if one  is able to break
\textit{P-MPWP} privacy, then one is also able to break it in
\textit{MPWP}.

\section{\uppercase{A Linear Dot Product Protocol}}
\label{sec:protocol}
\subsection{Overview with Three Players}
We first present in Figure~\ref{fig:dot3} our $DSDP_3$ protocol
(Distributed and Secure Dot-Product), for $3$ players.  The idea is
that Alice is interested in computing a dimension $3$ dot-product
$S=u^T \cdot v$, between her vector $u$ and a vector $v$ whose coefficients
are owned by different players. The other players send their
coefficients, encrypted, to Alice. Then she homomorphically multiplies
each one of these by her $u_i$ coefficients and masks the obtained
$u_iv_i$ by a random value $r_i$. Then the other players can decrypt
the resulting $u_iv_i+r_i$: with two unknowns $u_i$ and $r_i$ they are
not able to recover $v_i$. Finally the players enter a ring
computation of the overall sum before sending it to Alice. Then only,
Alice removes her random masks to recover the final dot-product. Since
at least two players have added $u_2v_2+u_3v_3$, there is at least two
unknowns for Alice, but a single equation.

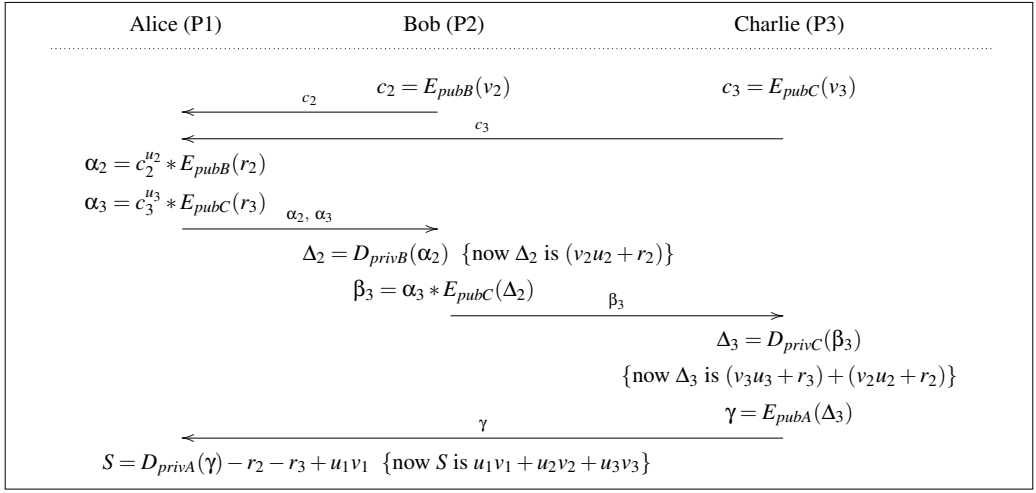
\begin{figure*}[htb]\center
\resizebox{0.85\textwidth}{!}{\framebox{\centerline{
$$\xymatrix@R=0pt@C=12pt{
& \text{Alice (P1)} & & \text{Bob (P2)} & & \text{Charlie (P3)} & \\
\hspace{45pt}\ar@{.}[rrrrrr]\hspace{-25pt}&&&&&&\\
& && &&&&\\
 & && c_2=E_{pubB}(v_2)&& c_3=E_{pubC}(v_3)&\\
& && \ar[ll]_{c_2} &&&\\
& && &&&&\\
& &&&& \ar[llll]_{c_3} &\\
 & \alpha_2=c_2^{u_2} *E_{pubB}(r_2)&&&&&\\
 & \alpha_3=c_3^{u_3} *E_{pubC}(r_3)&&&&&\\
& \ar[rr]^{\alpha_2,~\alpha_3} &&&&&\\
 &&&\hspace{-20pt}\Delta_2=D_{privB}(\alpha_2)~~\{\text{now}~\Delta_2~\text{is}~{(v_2u_2+r_2)}\}\hspace{-60pt}&&&\\
 &&&\beta_3=\alpha_3 *E_{pubC}(\Delta_2) &&&\\
& && \ar[rr]^{\beta_3} &&&\\
&  &&&& \Delta_3=D_{privC}(\beta_3)&\\
 &&&&& \{\text{now}~\Delta_3~\text{is}~{(v_3u_3+r_3)+(v_2u_2+r_2)}\}&\\
 &  &&&& \gamma=E_{pubA}(\Delta_3)&\\
& && &&\ar[llll]_{\gamma} &\\
 &S=D_{privA}(\gamma)-r_2-r_3+u_1v_1 
~~\{\text{now}~S~\text{is}~{u_1v_1+u_2v_2+u_3v_3}\}\hspace{-180pt}&&&&&\\
}$$
}%
}%
}%
\caption{$DSDP_3$: Secure dot product of vectors of size $3$ with a
  Paillier-like asymmetric cipher.}
    \label{fig:dot3}
\end{figure*}

We need that after several decryptions and re-encryptions, and removal
of the random values $r_i$, $S$ is exactly $\sum u_i
v_i$. The homomorphic Properties~(\ref{homo:add})
and~(\ref{homo:mul}) only guaranty that $D( \HAdd{\HMul{E(v_i)}{u_i}}{r_i}
)=v_iu_i+r_i\mod N_i$, for the modulo $N_i$ of the cryptosystem used
by player $P_i$. But then these values must be re-encrypted with
another player's cryptosystem, potentially with another
modulo. Finally Alice also must be able to remove the random values
and recover $S$ over $\Z$.  On the one hand, if players can share the
same modulo $M=N_i$ for the homomorphic properties then decryptions
and re-encryptions are naturally compatible.  This is possible for
instance in Benaloh's cipher.  On the other hand, in a Paillier-like
cipher, at the end of the protocol, Alice will actually recover
$S_4=((u_2v_2+r_2)\mod N_2 +u_3v_3+r_3) \mod N_3$.  He can remove
$r_3$, \emph{via} $S_3=S_4-r_3 \mod N_3$, but then
$S_3=((u_2v_2+r_2)\mod N_2 +u_3v_3) \mod N_3$. Now, if vectors
coefficients are bounded by say $B$, and if the third modulo is larger
than the second, $N_3>N_2+B^2$, the obtained value is actually the
exact value over the naturals: $S_3=(u_2v_2+r_2)\mod N_2+u_3v_3$. Then
Alice can remove the second random value, this time modulo $N_2$:
$S_2=(u_2v_2+u_3v_3)\mod N_2$, where now $N_2>2B^2$ suffices to
recover $S=S_2\in\N$.  We generalize this in the following section.

\subsection{General Protocol with $n$ Players}

We give the generalization $DSDP_n$, of the protocol of
Figure~\ref{fig:dot3} for $n$ players  in Algorithm~\ref{tn_protocol}
hereafter. 
For this protocol to be correct, we use the previously defined
bound~(\ref{hyp:bound}) on the players' private inputs.
\begin{algorithm}[htb]  
 \caption{$DSDP_n$ Protocol: Distributed and Secure Dot-Product of size~$n$}\label{tn_protocol}  
  \begin{algorithmic}[1]
    \REQUIRE $n\geq 3$ players, two vectors $U$ and $V$ such that
    $P_1$ knows complete vector $U$, and each players $P_i$
    knows component $v_i$ of $V$, for $i=1 \ldots n$;
    \REQUIRE $E_i$ (resp. $D_i$), encryption (resp. decryption)
    function of $P_i$, for $i=2 \ldots n$. 
    \ENSURE $P_1$ knows the dot-product $S=U^TV$. 
    \FOREND{$i=2 \ldots n$}{$\{P_{i} : c_{i}=E_i(v_i)$; $P_{i} \overset{c_i} \rightarrow P_{1}\}$}
    \FOR{$i=2 \ldots n$}
      \STATE $P_1: r_i \overset{\$}{\leftarrow}\Zpz{N_i}$
      \STATE $  P_{1} : \alpha_i = c_{i}^{u_i} * E_i(r_{i})$ so that $\alpha_i = E_i(u_iv_i+r_{i})$
    \ENDFOR
    \STATE $P_{1} \overset{\alpha_2}  \rightarrow P_{2}$ 
    \FOREND{$i=2 \ldots n-1$}{$  P_{1} : \overset{\alpha_{i+1}}  \rightarrow P_{i}$}
    \STATE $P_2$ : $\Delta_2 = D_2(\alpha_2) $ so that $\Delta_2= u_2v_2 + r_2$
    \STATE $P_2$ : $\beta_3 = \alpha_3 * E_3(\Delta_2) $ so that $\beta_3 = E_3(u_3v_3+r_3+\Delta_2)$; $P_2$ $\overset{\beta_3}  \rightarrow P_{3}$
    \FOR{$i=3 \ldots n-1$}
      \STATE\label{prot:delta} $P_i$ : $\Delta_i=
      D_i(\beta_{i})$ so that $\Delta_i=\sum_{k=2}^{i}  u_kv_k + r_k$ 
      \STATE\label{prot:beta} $P_i$ : $\beta_{i+1} = \alpha_{i+1} * E_{i+1}(\Delta_i) $ so that $\beta_{i+1} = E_{i+1}(u_{i+1}v_{i+1}+r_{i+1}+\Delta_i)$; $P_i$ $\overset{\beta_{i+1}}  \rightarrow P_{i+1}$
    \ENDFOR
    \STATE $P_n$ : $\Delta_n = D_{n}(\beta_{n})$; $P_n$ : $\gamma = E_{1}(\Delta_{n})$; $P_n \overset{\gamma}  \rightarrow P_{1}$
    \RETURN\label{prot:return} $P_{1} : S = D_1(\gamma) - \sum_{i=1}^{n-1} r_{i}  + u_1v_1$.
\end{algorithmic}  
\end{algorithm} 
Then, for $n$ players, there are two general cases:
First, if all the players share the same
  modulo $M=N_i$ for all~$i$ for the homomorphic properties, 
  then Alice can also use $M$ to remove the $r_i$. Then,
  to compute the correct value $S$, it is sufficient to satisfy the
  bound~(\ref{hyp:benaloh}).
Second, for a Paillier-like cipher, differently, the modulo of the
  homomorphic properties are distinct. We thus prove the following
  Lemma~\ref{lem:paillier}.
\begin{lemma}\label{lem:paillier}
Under the bound (\ref{hyp:bound}), and for any $r_i$,
let $M_2=(u_2v_2+r_2)\mod N_2$ and $M_{i}=(M_{i-1}+u_iv_i+r_i)\mod N_i$,
for $i=2 \ldots n-1$. Let also $S_{n+1}=M_n$ and $S_i=(S_{i+1}-r_i)\mod N_i$
for $i=n \ldots 2$. 
If we have:
\begin{equation}\label{hyp:paillier}
\begin{cases}
N_{i-1}+(n-i+1)B^2<N_{i},&~\text{for all}~i=3..n\\
(n-1)B^2<N_2
\end{cases}
\end{equation}
then $S_2=\sum_{i=2}^n u_iv_i\in\N$.
\end{lemma}
\begin{proof}
By induction, we first show that $S_i=M_{i-1}+\sum_{j=i}^n u_jv_j$, for $i=n..3$:
indeed $S_n=(M_n-r_n)\mod N_n=(M_{n-1}+u_nv_n)\mod N_n$.
But $M_{n-1}$ is modulo $N_{n-1}$, so
$(M_{n-1}+u_nv_n)<N_{n-1}+B^2$, and then~(\ref{hyp:paillier})
for $i=n$, ensures that $N_{n-1}+B^2<N_n$ and
$S_n=M_{n-1}+u_nv_n\in\N$.
Then, for $3\leq i<n$, $S_i=(S_{i+1}-r_i)\mod N_i=(M_i+\sum_{j=i+1}^n
u_jv_j-r_i)\mod N_i=(M_{i-1}+u_iv_i+r_i+\sum_{j=i+1}^n u_jv_j-r_i)\mod
N_i=(M_{i-1}+\sum_{j=i}^n u_jv_j)\mod N_i$, by induction.
But~(\ref{hyp:bound}) enforces that $M_{i-1}+\sum_{j=i}^n
u_jv_j<N_{i-1}+(n-i+1)B^2$ and~(\ref{hyp:paillier})
also ensures the latter is lower than $N_i$. Therefore
$S_i= M_{i-1}+\sum_{j=i}^n u_jv_j$ and the induction is proven.
Finally, $S_2=(S_3-r_2)\mod N_2=(M_2+\sum_{j=3}^n u_jv_j-r_2)\mod
N_2=(\sum_{j=2}^n u_jv_j)\mod N_2$. 
As $\sum_{j=2}^n u_jv_j<(n-1)B^2$, by~(\ref{hyp:paillier}) for $i=2$,
we have 
$S_2=\sum_{j=2}^n u_jv_j\in\N$.
\end{proof}
This shows that the $DSDP_n$ protocol of Algorithm~\ref{tn_protocol} can be
implemented with a Paillier-like underlying cryptosystem, provided
that the successive players have increasing modulo for their
public keys.
\begin{theorem} Under the bounds~(\ref{hyp:bound}),
and under Hypothesis~(\ref{hyp:benaloh}) with a shared modulus underlying
cipher, or under Hypothesis~(\ref{hyp:paillier}) with a Paillier-like
underlying cipher, the $DSDP_n$ protocol of Algorithm~\ref{tn_protocol} is
correct.
It requires $\bigO{n}$ communications
and $\bigO{n}$ encryption and decryption operations.
\end{theorem}
\begin{proof}
  First, each player sends his
  ciphered entry to $P_1$, then homomorphically added to random values, $r_i$.  
  Then, $P_i$ ($i \ge 2$) deciphers the
  message received by $P_{i-1}$ into $\Delta_i$. 
  By induction, we obtain $\Delta_i = \sum_{k=2}^{i} u_kv_k + r_k $. 
  This value is then re-enciphered with next player's key
  and the next player share is homomorphically added. 
  Finally, $P_1$ just has to remove all
  the added randomness to obtain  $S = \Delta_n - \sum_{i=2}^{n}
  r_{i} + u_1v_1= \sum_{i=1}^{n}u_iv_i $.
For the complexity, the protocol needs $n-1$
encryptions and communications for the $c_i$; $2(n-1)$ homomorphic
operations on ciphers and $n-1$ communications for the $\alpha_i$;
$n-1$ decryptions for the $\Delta_i$; $n-1$ encryptions, homomorphic
operations and communications for the $\beta_i$; and finally one
encryption and one communication for $\gamma$. Then $P_1$ needs
$\bigO{n}$ operations to recover~$S$.
\end{proof}

\section{\uppercase{Security of $DSDP$} }
\label{sec:security}

 We study the security of $DSDP_n$ using both mathematical proofs and
 automated verifications.  We first demonstrate the security of the
 protocol for {\em semi-honest} adversaries.  Then we incrementally
 build its security helped by attacks found by \proverif{}, an automatic
 verification tool for cryptographic protocols. 

\subsection{Security Proofs}\label{sec:proofs}

The standard security definition in MPC models~\cite{Lindell:09}
covers actually many security issues, such as correctness, inputs
independence, privacy, etc. 
We first prove that under this settings, computation of the dot product
is safe.
\begin{lemma}
For $n \ge 3$, the output obtained after computing a dot product where
one player owns complete vector $U$, and where each coefficient $v_i$ of
the second vector $V$ is owned by the player $P_i$, is
safe.
\end{lemma}
\begin{proof}
After executing $DSDP_n$ with $n \ge 3$, $P_1$ received the dot
product of $U$ and $V$. Therefore, it owns only one equation
containing $(n-1)$ unknown values (coefficients from $v_2$ to
$v_n$). Then, he cannot deduce other players' private inputs.
\end{proof}

Then, proving the security relies on a comparison
between a real-world protocol execution and an ideal one. The latter
involves an hypothetical trusted third party ($TTP$) which, knowing
only the players' private inputs, returns the correct result to the correct
players. The protocol is considered secure if the players'
views in the ideal case cannot be distinguished from the real
ones. Views of a player $P_i$ (denoted $View_{P_i}$) are defined as
distributions containing: the players' 
inputs (including random values), the messages received during a protocol
execution and the outputs.
 The construction of the corrupted players' view in the
ideal world is made by an algorithm called $Simulator$.

\begin{definition}
In the presence of a set $C$ of semi-honest adversaries
with inputs set $X_C$, a protocol $\Pi$ securely 
computes $f: ([0,1]^*)^m \rightarrow
([0,1]^*)^m$ (and $f_C$ denotes the outputs of $f$ for each
adversaries in $C$)  if there exists a
probabilistic polynomial-time algorithm $Sim$, such that:
$\{Sim(C,\{X_C\}, f_C(X))\}_{X \in ([0,1]^*)^m}$ is
computationally indistinguishable from $\{C,\{View_{P_i}^\Pi\}_{P_i\in C}\}$.
\end{definition}
For $DSDP_n$, it is secure only if $C$ is reduced to a singleton,
\textit{i.e.} if only one player is corrupted. 
\begin{lemma}
By assuming the semantic security of the cryptosystem $E$, for
$n \ge 3$, $DSPD_n$ is secure against one semi-honest adversary.  
\end{lemma}
\begin{proof}
We assume that the underlying cryptosystem $E$ is semantically secure
(IND-CPA secure). First, we suppose that only $P_1$ is corrupted.  His
view, in a real execution of the protocol, is  
$View_{P_1} = \{U,R,\gamma, S, A, B,C \}$, where  $U = \{u_i\}_{1 \le
  i \le n}$, $R = \{r_i\}_{1 \le i \le n}$, $ A
= \{\alpha_i\}_{2 \le i \le n}$, $B = \{\beta_i\}_{3 \le i \le
n-1}$ and $C = \{c_i\}_{2 \le i \le n}$. Now, $Sim_1$ is the simulator
for $P_1$ in the ideal case, where a simulated value $x$ is denoted
$x'$: by definition,  $P_1$'s private entries (vectors $U$ and $R$) are
directly accessible to $Sim_1$, along with the output $S$, sent by the
$TTP$. $Sim_1$ starts by generating $n-2$
random values, and then ciphers them using the corresponding public
keys: this simulates the $c_i'$ values. Then, using the provided $r_i$
and $u_i$ with the 
associated $c_i'$ and $P_i$'s public key, $Sim_1$ computes:
$\alpha_i' = c_i'^{u_i}*E_i(r_i), 2 \le i \le n$. Next, the simulation of
$B'$ is done by ciphering random values with the appropriate public key.
The $\gamma'$ value is computed
using $R$ along with the protocol output $S$: $\gamma' =
E_1(S + \sum_i^{n-2}r_i + u_1v_1)$. In the end, the simulator view is
$View_{Sim_1} = \{U,R,\gamma', S, A', B', C' \}$. If an adversary is
able to distinguish any ciphered values (\emph{e.g.} $C'$ from $C$ and thus
$A'$ from $A$), hence 
he is able to break the semantic security of the
underlying cryptographic protocol. This is assumed impossible. Moreover,
since the remaining values are computed as in a real execution, $P_1$
is not able to distinguish $View_{P_1}$ from $View_{Sim_1} $.
Second, we suppose that a player $P_i, i \ge 2$ is corrupted and denote
by $Sim_i$ the simulator in this case. Since
the role played by each participant is generic, (except for $P_n$,
which only differs by his computation of $\gamma$
instead of $\beta_{n+1}$), the simulators are easily adaptable. During a real
protocol execution, the view of $P_i$ is $View_{P_i} =
\{v_i, A, B, C, \gamma, \Delta_i\}$. 
Simulating the values also known to
$P_1$ is similar, up to the used keys. Hence,
the simulation of $A'$, $B'$, $\gamma'$, $C'$ (except $c_i$) is made
by ciphering random values 
using the adequate public key. $c_i$ is ciphered using $v_i$ and
the public key of $P_i$. For $\Delta_i'$, the simulator $Sim_i$
has to forward the random value previously chosen to be ciphered as
$\alpha_i$. Indistinguishability is based on the semantic
security of $E$ (for $A$, $B$, $C$ and $\gamma$) and on
the randomness added by $P_1$ (and thus unknown by $P_i$). Then,
$\Delta_i'$ is computationally indistinguishable from the real
$\Delta_i$.  Hence,  $View_{P_i}$ and  $View_{S_i}$ are
indistinguishable and $DSDP_n$ is secure against one semi-honest adversary. 
\end{proof}

\subsection{Automated Verification}\label{sec:proverif}

Alongside mathematical proofs, we use an automatic protocol
verification tool to analyze the security of the protocol.  Among
existing tools, we use \proverif{}~\cite{blanchet01,ProVerif}.
It allows users to add their
own equational theories to model a large class of  protocols.  In our
case, we model properties of the underlying cryptosystem including
addition and multiplication.  Sadly, verification of protocol in
presence of homomorphic function over abelian groups theory has been
proven undecidable \cite{Delaune:2006:URA:1221599.1222026}.  Moreover,
as showed in \cite{LP15}, some equational theories such as
Exclusive-Or can already outspace the tool's capacities. Thus we have
to provide adapted equational theories to be able to obtain results
with the tool.  We modeled the application of Pailler's or shared
modulus encryption properties on $\alpha_{i}$ messages that Bob
receives as follows: \newcounter{myenumi}
\begin{compactenum}\renewcommand {\theenumi}{(\roman{enumi})}
    \item $\forall u,v,r,k, \; bob(E_k(r), u, E_k(v)) = E_k(uv + r)$
\setcounter{myenumi}{\value{enumi}}
\end{compactenum}
This property allows Bob to obtain $u_{2}v_{2} + r_{2}$ from $\alpha_{2}$.
This also allows an intruder to simulate such calculus and impersonate Bob.
We also model:
\begin{compactenum}\renewcommand {\theenumi}{(\roman{enumi})}\setcounter{enumi}{\themyenumi}
    \item $\beta_{3}$ by $\forall u,v,r,x,y,z,k, \;
        charlie(E_k(uv + r), E_k(xy + z)) = E_k(uv + xy + r + z)$
    \item $\beta_{4}$ by $\forall u,v,r,x,y,z,a,b,c,k, \;
        dave(E_k(uv + xy + r + z), E_k(ab + c)) = E_k(uv + xy + ab + r + z + c)$
\end{compactenum}
In the following, we use \proverif{} to prove the security of our
protocols under the abstraction of the functionalities given in our
equational theory. \proverif{} discovers some attacks in presence of
active intruder. We then propose some countermeasures. 
The limits of \proverif{} are  reached and it does not
terminate. The associated source files are
available in a web-site: \textit{http://matmuldistrib.forge.imag.fr}

\paragraph{Analysis in case of a passive  adversary.}
Using these equational theories on the protocol described in
Figure~\ref{fig:dot3}, we verify it in presence of  a passive intruder.  Such
adversary is able to observe all the traffic of the protocol and tries
to deduce secret information of the messages.  This corresponds to a
\proverif{} intruder that only listens to the network and does not
send any message.  By default, this intruder does not possess the
private key of any agent and thus does not belong to the protocol. To
model a {\em semi-honest} adversary as defined in
Section~\ref{sec:background}, we  just give secret keys of honest
participants to the passive intruder knowledge in \proverif{}. Then
the tool proves that all secret terms cannot be learn by the intruder
for any combinations of leaked key.  This confirms the proofs given in
Section~\ref{sec:proofs} against the {\em semi-honest} adversaries.

\paragraph{Analysis in case of malicious adversary.}
The {\em malicious} adversary described in Section \ref{sec:background} is
an active intruder that controls the network and knows a private key of a
compromised honest participant.  Modeling this adversary in \proverif, we are
able to spot the two following attacks and give some  countermeasures:

\begin{compactenum}
\item[(i)]
Only the key of Alice is compromised and the countermeasure uses proofs of
knowledge. 
\item[(ii)]
 Only the key of Charlie is compromised and the countermeasure
uses signatures. 
\end{compactenum}
In the rest of the section, we present these two
points. In the Section~\ref{sec:randomrings}, we also give a solution called
\emph{random ring} for the case where both keys of Alice and Charlie are
compromised.

\emph{(i) The key of Alice is compromised.} An attack on the secrecy of $v_{2}$,
the secret generated by Bob, is then presented in Figure~\ref{fig:atk_v2}.
\begin{figure}[tb]
    \resizebox{\columnwidth}{!}{
\input{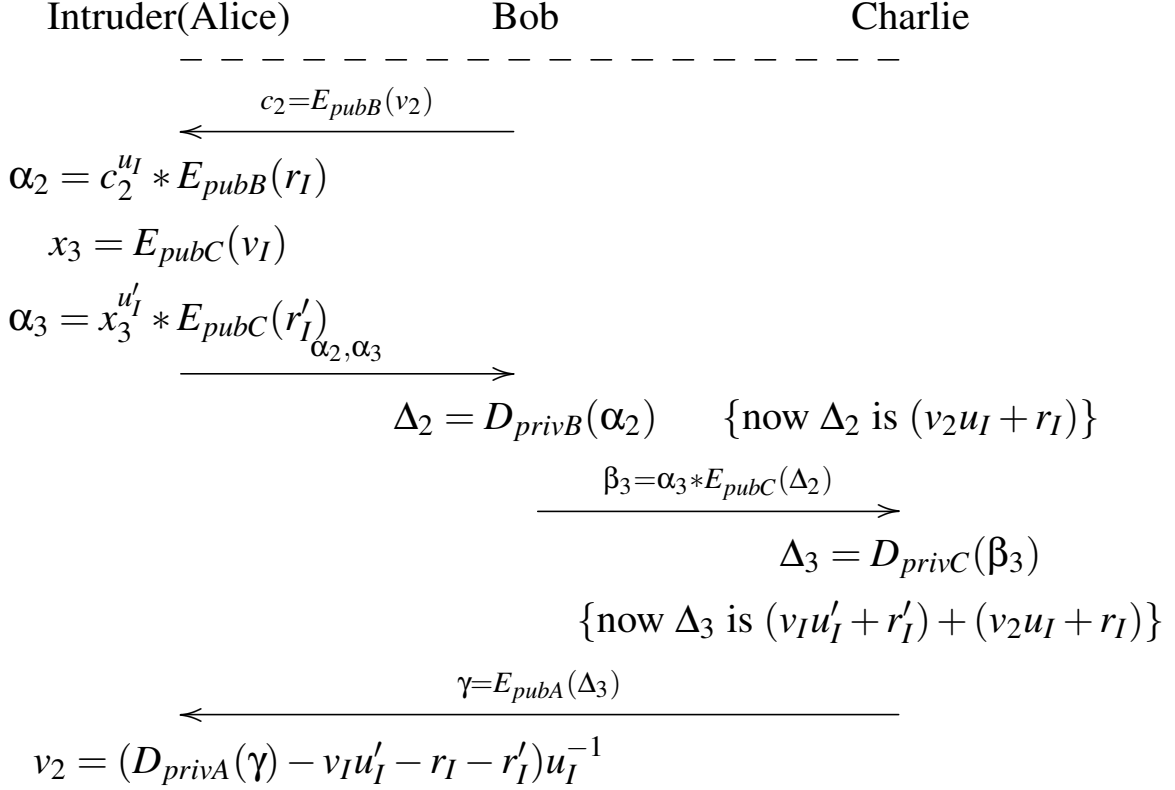}
}
\caption{Attack on the secrecy of $v_{2}$}\label{fig:atk_v2}
\end{figure}

 The malicious adversary usurps Alice and replaces all the $\alpha_{i}$ messages,
 arriving from the other agents, with one message she generated, except
 one message, denoted $c_{2}$ in Figure~\ref{fig:atk_v2}.  He lets the
 protocol end normally and obtains a term where only $v_{2}$ is
 unknown. He learns $v_2$.  If the key of Alice ($P_1$) is compromised,
 \proverif{}  also finds an attack on any of the other players secrecy.
 Suppose, w.l.o.g, that $P_2$ is the target, $P_1$ replaces each $\alpha_i$
 except $\alpha_2$ by ciphers $E_i(x_i)$ where $x_i$ are known to him.
 $x_i=0$ could do for instance ($x_i=0v_i+r_i$ also), since after
 completion of the protocol, $P_1$  learns $u_2v_2+r_2+\sum_{i=3}^n
 x_i$, where the $u_i$ and $r_i$ are known to him. Therefore, $P_1$
  learns $v_2$. Note also that similarly, for instance,
 $\alpha_2=1v_2+0$ and $x_3=v_3$ could also reveal $v_2$ to~$P_3$.
{\it Counter measure:} this attack, and more generally attacks on the
form of the $\alpha_i$ can be counteracted by zero-knowledge proofs of
knowledge.  $P_1$ has to prove to the other players that $\alpha_i$ is a
non trivial affine transform of their secret $v_i$. For this we use a
variant of a proof of knowledge of a discrete logarithm~\cite{Chaum86}
given in Figure~\ref{fig:pdl}.
\begin{figure}[htb]\center
\resizebox{\columnwidth}{!}{
\input{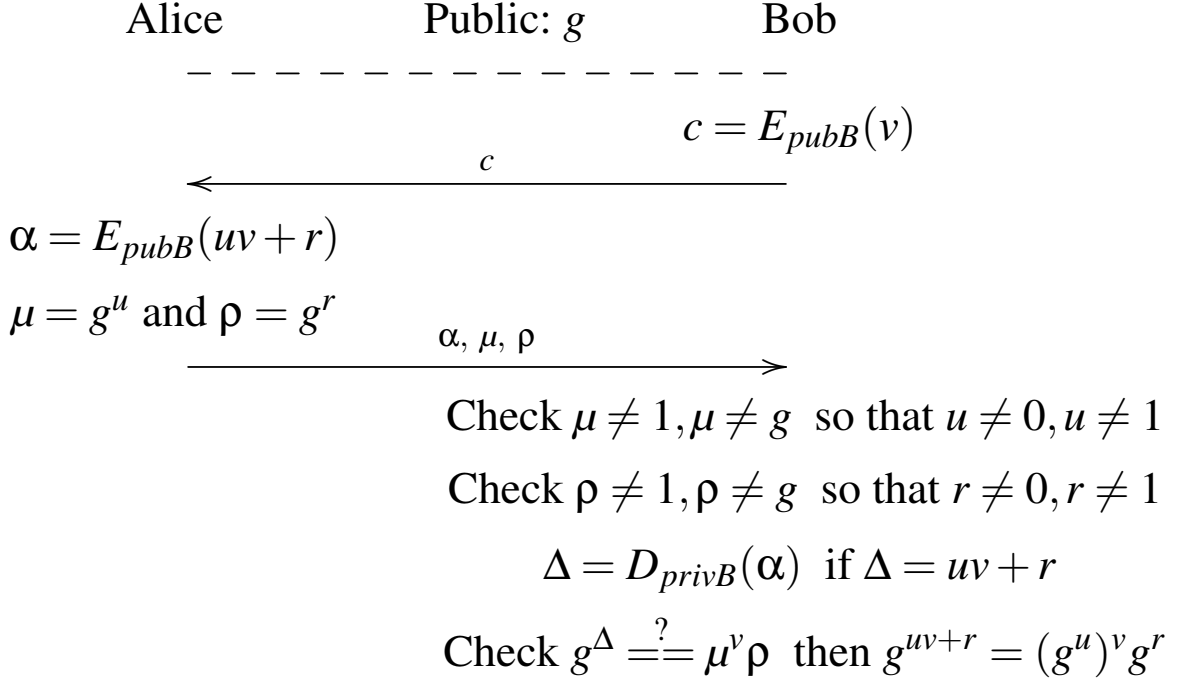}
}
\caption{Proof of a non trivial affine transform}\label{fig:pdl}
\end{figure}

In the Protocol~\ref{tn_protocol}, this proof of a non trivial affine
transform applies as is to $\alpha_2$ with $\mu_2=g^{u_2}$,
$\rho_2=g^{r_2}$ so that the check of $P_2$ is
$\delta_2=g^{\Delta_2}\stackrel{?}{==}\mu_2^{v_2}\rho_2$.
Differently, for the subsequent players, the
$\delta_{i-1}=g^{\Delta_{i-1}}$ used to test must be forwarded: indeed
the subsequent players have to check in line~\ref{prot:delta} that
$\Delta_i=u_iv_i+r_i+\Delta_{i-1}$.  Thus with $P_1$ providing
$\mu_i=g^{u_i}$, $\rho_i=g^{r_i}$ and $P_{i-1}$ providing
$\delta_{i-1}$, the check of player $P_i$ ends with
$\delta_i=g^{\Delta_i}\stackrel{?}{==}\mu_i^{v_i}\rho_i\delta_{i-1}$.
As for proofs of knowledge of discrete logarithm, secrecy of our proof
of non trivial affine transform is guaranteed as long as the discrete
logarithm is difficult.  The overhead in the protocol, in terms of
communications, is  to triple the size of the messages from
$P_1$ to $P_i$, with $\alpha_i$ growing to $(\alpha_i,\mu_i,\rho_i)$,
and to double the size of the messages from $P_i$ to $P_{i+1}$, with
$\beta_i$ growing to $(\beta_i,\delta_i)$.  In terms of computations,
it is also a neglectible linear global overhead.

\emph{(ii) The key of Charlie is compromised.} There \proverif{} finds
another attack on the secrecy of $v_{2}$.  This time the key of
Charlie is compromised and the malicious adversary blocks all
communications to and from Alice who is honest.  The adversary
performs the same manipulation on the $\alpha_i$ terms which are
directly sent to Bob.  Thus, this attack becomes feasible since the
adversary knows the terms $u_{2}$, $u_{3}$, $r_{2}$, $r_{3}$ and
$v_{3}$ that he generated and $\Delta_{3} = (v_2u_2+r_2)+(v_3u_3+r_3)$
using the private key of Charlie.  Such an attack relies on the fact
that Bob has no way to verify if the message he receives from Alice
has really been sent by Alice.  This can be avoided using
cryptographic signatures.%
This attack can be generalized to any number of participants.
The attack needs the adversary to know the key of Alice (since she is the only
one to know the $u_{i}$ and $r_{i}$ values thanks to the signatures).
Then, to obtain the secret value of a participant $P_{i}$, the key of
participants $P_{i-1}$ and $P_{i+1}$ are also needed:

\begin{compactenum}\renewcommand {\theenumi}{(\roman{enumi})}
    \item $P_{i-1}$ knows $\Delta_{i-1} = (u_{2}v_{2} + ... + u_{i-1}v_{i-1}
        + r_{2} + ... + r_{i-1})$.
    \item $P_{i+1}$ knows $\Delta_{i+1} = (u_{2}v_{2} + ... + u_{i-1}v_{i-1}
        + u_{i}v_{i} + u_{i+1}v_{i+1} + r_{2} + ... + r_{i-1} + r_{i}
        + r_{i+1})$.
\end{compactenum}

Thus, by simplifying $\Delta_{i-1}$ and $\Delta_{i+1}$, the malicious adversary
obtains $u_{i}v_{i} + u_{i+1}v_{i+1} + r_{i} + r_{i+1}$ where he can remove
$u_{i+1}$, $v_{i+1}$, $r_{i}$, $r_{i+1}$ and $u_{i}$ to obtain $v_{i}$.
For more than three participants, we see in
Section~\ref{sec:randomrings} that these kinds of threats can be
diminished if the protocol is replayed several times in random orders.

\section{\uppercase{Parallel Approach}}\label{sec:matmul}
In order to speed up the overall process, we show that we can cut each
dot-product into blocks of $2$ or $3$ coefficients.  On the one hand,
the overall volume of communications is unchanged, while the number of
rounds is reduced from $n$ to a maximum of $5$.  On the other hand,
semantic security is dropped, but we will see at the end of this
section that by simply repeating the protocol with a wiretap mask it
is possible to make the probability of breaking the protocol
negligible.

An application of the $DSDP_n$ protocol is the computation of matrix
multiplication.
In this case, instead of knowing 
one vector, each player $P_i$ owns two rows, $A_i$ and $B_i$, one of
each $n\times{}n$ matrices $A$ and $B$. At the end, each $P_i$ learns a row
$C_i$ of the matrix $C=AB$. 
In order to compute the matrix product, it
is therefore natural to parallelize $DSDP_n$:
each dot-product is cut into blocks of $2$ or $3$ coefficients.
Indeed, scalar product
between three players (resp. four) involves two (resp. three) new
coefficients in addition to the ones already known by $P_i$.
For $P_1$, the idea is to call $DSDP_3$ on the coefficients
$u_1,v_1$ and $u_2,u_3$ of $P_1$,
and $v_2,v_3$ of $P_2$ and $P_3$. Then $P_1$ knows
$s=u_1v_1+u_2v_2+u_3v_3$. 
$P_1$ can then continue the protocol with $P_4$ and $P_5$, using 
$(s,1)$ as his first coefficient  and $u_4,u_5$ to be combined with
$v_4,v_5$, etc. 
$P_1$ can also launch the computations in parallel. Then $P_1$ 
adds his share $u_1v_1$ only after all the computations.
For this it is sufficient to modify line~\ref{prot:return} of
$DSDP_n$ as: 
$P_{1} : S = D_1(\gamma) - \sum_{i=1}^{n-1} r_{i} $.
This is given as the $ESDP_n$ protocol variant in
Algorithm~\ref{en_protocol}.
\begin{algorithm}[htb]  
 \caption{$ESDP_n$ Protocol: External Secure Dot-Product of
   size~$n$}\label{en_protocol} 
  \begin{algorithmic}
\REQUIRE $n+1$ players, $P_1$ knows a coefficient vector $U\in\F^n$,
each $P_i$ knows components $v_{i-1}$ of $V\in\F^n$, for $i=2 \ldots n+1$.
\ENSURE $P_1$ knows $S=U^TV$.
\RETURN $DSDP_{n+1}(P_1\ldots{}P_{n+1},[0,U],[0,V])$.
\end{algorithmic}
\end{algorithm}
\subsection{Partition in Pairs or Triples}
Depending on the parity of $n$, and since $gcd(2,3)=~1$, calls to
$ESDP_2$ and $ESDP_3$ are sufficient to cover all possible dot-product
cases, as shown in protocol $PDSMM_n$ of Algorithm~\ref{mn_protocol}.
The protocol is cut in two parts. The loop allows us to go all
over coefficients by block of size $2$. 
In the case where $n$ is even, a block of $3$ coefficients 
is treated with an instance of
$ESDP_3$. 
In terms of efficiency and depending on the parity of $n$,
$ESDP_2$ is  called $\frac{n-1}{2}$ or $\frac{n}{2}-2$ times, and $ESDP_3$ is  called $0$ or $1$ times.

\begin{algorithm}[htb]\caption{$PDSMM_n$ Protocol: Parallel Distributed
    and Secure Matrix Multiplication}\label{mn_protocol}
  \begin{algorithmic}[1]
    \REQUIRE $n$ players, each player $P_i$ knows rows $A_i$ and $B_i$
    of two $n\times{}n$ matrices $A$, $B$.
    \ENSURE Each player $P_i$ knows row $i$ of $C=AB$.
    \FOR{Each row: i=1 \ldots n}
    \FOR{Each column: j=1 \ldots n}
    \STATE\label{sinit} $s\leftarrow a_{i,i} b_{i,j}$
    \IF{$n$ is even}
    \STATE $k_1\leftarrow (i-1)\mod n+1$; $k_2\leftarrow (i-2)\mod n+1$; $k_3\leftarrow (i-3)\mod n+1$; 
    \STATE $s \leftarrow s + ESDP_3(P_i$, $[P_{k_3},P_{k_2},P_{k_1}]$,
    $[a_{i,k_3},a_{i,k_2},a_{i,k_1}]$, $[b_{k_3,j},b_{k_2,j},b_{k_1,j}])$\\
    \STATE $t\leftarrow\frac{n-4}{2}$
    \ELSE
    \STATE $t\leftarrow\frac{n-1}{2}$
    \ENDIF
    \FOR{$h=1 \ldots t$} 
    \STATE $k_1\leftarrow (i+2h-1)\mod n+1$; $k_2\leftarrow (i+2h)\mod n+1$; 
    \STATE $s \leftarrow s + ESDP_2(P_i$, $[P_{k_1},P_{k_2}]$,
    $[a_{i,k_1},a_{i,k_2}]$, $[b_{k_1,j},b_{k_2,j}])$\\
    \ENDFOR
    \STATE $c_{i,j}\leftarrow s$
    \ENDFOR
    \ENDFOR
  \end{algorithmic}  
\end{algorithm}

\begin{theorem}
  The $PDSMM_n$ Protocol in Algorithm~\ref{mn_protocol} is correct.
  It runs in less than $5$ parallel communication rounds.
\end{theorem}

\begin{proof}
Correctness means that at the end, each $P_i$ has learnt row $C_i$ of
$C=AB$. Since the protocol is applied on each rows and columns,
let us show that for a row $i$ and a column $j$, Algorithm~\ref{mn_protocol} 
gives the coefficient $c_{ij}$ such that 
$c_{ij} = \sum_{k=1}^na_{ik}*b_{kj}$.
First, the $k_i$ coefficients are just the values $1 \ldots (i-1)$ and
$(i+1) \ldots n$ in order. 
Then, the result of any $ESDP_2$ step is
$a_{i,k_1}b_{k_1,j}+a_{i,k_2}b_{k_2,j}$ and the result of the
potential $ESDP_3$ step is
$a_{i,k_3}b_{k_3,j}+a_{i,k_2}b_{k_2,j}+a_{i,k_1}b_{k_1,j}$.
Therefore accumulating them in addition of $a_{i,i}*b_{i,j}$ produces as expected 
$c_{ij} = \sum_{k=1}^na_{ik}*b_{kj}$.

Now for the number of rounds, for all $i$ and $j$, all the $ESDP$
calls are independent.
Therefore, if each player can simultaneously send and receive multiple
data we have that:
in parallel, $ESDP_2$, like $DSDP_3$ in Figure~\ref{fig:dot3},
requires $4$ rounds with a constant number of operations: one round
for the $c_i$, one round for the $\alpha_i$, one 
round for $\beta_3$ and one round for $\gamma$. 
As shown in Algorithm~\ref{tn_protocol}, $ESDP_3$, like $DSDP_4$, requires
only a single additional round for $\beta_4$.
\end{proof}

\subsection{Random Ring Order Mitigation}\label{sec:randomrings}
We have previously seen that if
the first player of a dot-product cooperates with the third one she
can always recover the second player private value. If
the first player cooperates with two well placed players she can
recover the private value of a player in between.  In the trust
evaluation setting every malicious player plays the role of the first
player in its row and therefore as soon as there is a collaboration,
there is a risk of leakage.  To mitigate this cooperation risk, our
idea is to repeat the dot product protocol in random orders, except
for the first player.  To access a given private value, the malicious
adversaries  have to be well placed in \emph{every} occurrence of the
protocol. Therefore if their placement is chosen uniformly at random
the probability that they recover some private value diminishes with
the number of occurrences.  In practice, they  use a pseudo, but
unpredictable, random generator to decide their placement: as each of
them has to know their placement, they can for instance use a
cryptographic hash function seeded with the alphabetical list of the
players distinguished names, with the date of the day and with random values
published by each of the players.  We
detail the overall procedure only for one dot-product, within the
$PDSMM_n$ protocol. Each player except the first one
masks his coefficient $v$ as in a simple wiretap
channel~\cite{Ozarov:1984:wiretapII}, as sketched in
Algorithm~\ref{alg:wiretap}.
\begin{algorithm}[htb]
  \caption{Wiretap repetition of the dot-product}\label{alg:wiretap}
  \begin{algorithmic}[1]
    \STATE The players agree on  $d$  occurrences.
    \STATE Each player computes his placement order in each occurrence of
    the protocol from the cryptographic hash function.
\makeatletter\renewcommand{\ALC@linenodelimiter}{a:}\makeatother
    \STATE  With a shared modulus cryptosystem, the players should share a
      common modulo $M$ satisfying Hypothesis~(\ref{hyp:benaloh}).
      In the first occurrence, each player $P_j$ then masks his private input
      coefficient $v_j$ with $d-1$ random values $\lambda_{j,i}\in\Zpz{M}$:
$%
        v_j-\sum_{i=2}^d \lambda_{j,i}.
$%
\addtocounter{ALC@line}{-1}%
\makeatletter\renewcommand{\ALC@linenodelimiter}{b:}\makeatother%
   \STATE With a Paillier-like cryptosystem, the players choose their
      moduli according to Hypothesis~(\ref{hyp:paillier}), where $B^2$ is
      replaced by $dB^2$, in groups of size $n=4$ (the
        requirements of (\ref{hyp:paillier}) on the moduli
        are somewhat sequential, but  can be satisfied independently if
        each modulo is chosen in a distinct interval larger
        than~$3dB^2$).
      Then, in the first occurrence, each player $P_j$ masks his
      private input coefficient $v_j$ with $d-1$ random values
      $0\leq\lambda_{j,i}<B$: 
$%
        v_j+\sum_{i=2}^d (B-\lambda_{j,i})<dB. 
$%
\makeatletter\renewcommand{\ALC@linenodelimiter}{:}\makeatother
    \STATE Then for each subsequent occurrence, each player replaces its
    coefficient by one of the $\lambda_{j,i}$.
    \STATE In the end, the first player has gathered $d$ dot-products and
    just needs to sum them in order to recover the correct one.
  \end{algorithmic}
\end{algorithm}

\begin{theorem}
Algorithm~\ref{alg:wiretap} correctly allows the first player to compute the
dot-product.
\end{theorem}
\begin{proof}
First, in a shared modulus setting, after the first occurrence, Alice
  ($P_1$) gets $S_1=\sum_{j=2}^n u_j\left(v_j-\sum_{i=2}^d
  \lambda_{j,i}\right)$.  Then in the following occurrences, Alice
  gets $S_i=\sum_{j=2}^n u_j\lambda_{j,i}$.  Finally she computes
  $\sum_{i=1}^d S_i = \sum_{j=2}^n u_jv_j$.
Second, similarly, in a Paillier-like setting, after the first
  occurrence, Alice recovers $S_1=\sum_{j=2}^n
  u_j\left(v_j+\sum_{i=2}^d (B-\lambda_{j,i})\right)$.  Then in the
  following occurrences, Alice gets $S_i=\sum_{j=2}^n
  u_j\lambda_{j,i}$.  Finally she computes $\sum_{i=1}^d S_i -
  (d-1)B(\sum_{j=2}^n u_j) = \sum_{j=2}^n
  u_j(v_j+(d-1)B)-(d-1)Bu_j=\sum_{j=2}^n u_jv_j$.
\end{proof}

 We give now the probability of avoiding
attacks in the case when $n=2t+1$, but
the probability in the even case should be close.

\begin{theorem}\label{thm:wiretap:average}
Consider $n=2t+1$ players, grouped by $3$, of which $k\leq{}n-2$ are
malicious and cooperating, including the first one Alice.  Then, it is
on average sufficient to run Algorithm~\ref{alg:wiretap} with
$d\leq{}2\ln\left(\min\{k-1,n-k,\frac{n-1}{2}\}\right)\left(1+\frac{k-1}{n-k-1}\right)$
occurrences, to prevent the malicious players from recovering any
private input of the non malicious ones.
\end{theorem}
\begin{proof}
The idea is that for a given private input of a non malicious player
Bob, to be revealed to Alice, Bob needs to be placed between
cooperating malicious adversaries {\em at each occurrence} of the
protocol. If there is only one non malicious player, then nothing can
be done to protect him.  If there is $2$ non malicious, they are safe
if they are together one time, this happens with probability
$\frac{1}{n-2}$, and thus on average after $n-2$ occurrences.
Otherwise, $PDSMM_n$ uses $t=\frac{n-1}{2}$ groups of $3$, including
Alice.  Thus, each time a group is formed with one malicious and one
non malicious other players, Alice can learn the private value of the
non malicious player.  Now, after any occurrence, the number $a$ of
attacked players is less than the number of malicious players minus
$1$ (for Alice) and obviously less than the number of non malicious
players: $0\leq a<\min\{k-1,n-k\}$.  Thus let $b=k-1-a$ and $c=n-k-a$.
In the next occurrence, the probability of saving at least one more
non malicious is
$\frac{a(a-1+c)(n-3)!}{(n-1)!}\frac{n-1}{2}=\frac{a(a-1+c)}{2(n-2)}=\frac{a(n-k-1)}{2(n-2)}$,
so that the average number of occurrences to realize this is
$\E_{n,k}(a)=\frac{2(n-2)}{a(n-k-1)}$.  Thus, $T_{n,k}(a)$, the
average number of occurrences to save all the non malicious players,
satisfies $T_{n,k}(a)\leq{} \E_{n,k}(a)+T_{n,k}(a-1)\leq{}\sum_{i=a}^3
E_{n,k}(i)+T_{n,k}(2)=(\sum_{i=a}^3
\frac{1}{i})\frac{2(n-2)}{n-k-1}+T_{n,k}(2)$.  With $2$ attacked and
$c$ saved, $T_{n,k=n-c-2}(2)=\frac{n-2}{c+1}$ so that
$T_{n,k}(a)\leq{}(H_a-\frac{3}{2})\frac{2(n-2)}{n-k-1}+\frac{n-2}{n-k-1}$,
where bounds on the Harmonic numbers give $H_a\leq\ln{a}$ (see,
\emph{e.g.}, \cite{Batir:2011:harmonic}) and since $a\leq k-1$ and
$a\leq n-k$, this shows also that $2a\leq n-1$.  Therefore,
$T_{n,k}(a)\leq
2\ln\left(\min\{k-1,n-k,\frac{n-1}{2}\}\right)\frac{n-2}{n-k-1}$.
\end{proof}
For instance, if $k$, the number of malicious insiders, is less than
the number of non malicious
ones, the number of repetitions sufficient to prevent any attack is
on average bounded by \bigO{\log{k}}. 
To guaranty a probability of failure less than $\epsilon$, one needs
to consider also the worst case. There, we can have $k=n-2$ malicious
adversaries and the number of repetitions can grow to $n\ln(1/\epsilon)$:
\begin{proposition}\label{prop:wiretap:worst}
With $n=2t+1$, the number $d$ of random ring repetitions of
Algorithm~\ref{alg:wiretap} 
to make the probability of breaking the protocol lower than $\epsilon$
satisfies $d<n\ln(1/\epsilon)$ in the worst case.
\end{proposition}
\begin{proof}
There are at least $2$ non-malicious players, otherwise the
dot-product reveals the secrets in any case.
Any given non-malicious player is safe from any attacks if in at least
one repetition he was paired with another non-malicious player.
In the worst case, $k=n-2$ players are malicious and the latter
event arises with probability $(1-\frac{1}{n-1})^d$ for $d$
repetitions. 
If $d\geq{}n\left(\ln\left(\epsilon^{-1}\right)\right)$,
then
$d>(n-1)(-\ln{\epsilon})>\frac{\ln\epsilon}{\ln\left(1-\frac{1}{n-1}\right)}$,
which shows that $(1-\frac{1}{n-1})^d<\epsilon$.
\end{proof}
Overall, the wiretap variant of Algorithm~\ref{alg:wiretap} can
guaranty any security, at the cost of repeating the protocol. 
As the number of repetitions is fixed at the beginning by all the
players, all these repetitions can occur in parallel. Therefore, the
overall volume of communication is multiplied by the number of
repetitions, while the number of rounds remains constant.
This is summarized in Table~\ref{tab:comcomp} and
Figure~\ref{fig:timcomp}, for the average
(Theorem~\ref{thm:wiretap:average}) and worst
(Proposition~\ref{prop:wiretap:worst}) cases of Algorithm~\ref{alg:wiretap},
and where the protocols of the previous sections are also compared.

\begin{table}[ht]\small
\caption{Communication complexities}\label{tab:comcomp}
\centering{
\begingroup
\setlength{\tabcolsep}{0.5em} %
\renewcommand*{\arraystretch}{1.2} %
\begin{tabular}{|l||c|c|c|}
\hline
Protocol & Volume & Rounds & Paillier\\
\hline
MPWP&\bigO{n^3}& \bigO{n} & \xmark\\
P-MPWP~(\S~\ref{sec:p-mpwp})& $n^{2+o(1)}$& \bigO{n} & \cmark \\
Alg.~\ref{alg:wiretap} (Wiretap) & $n^{2+o(1)}\ln\left(\frac{1}{\epsilon}\right)$ & 5 & \cmark \\
\hline
Alg.~\ref{tn_protocol} (DSDP$_n$)& $n^{1+o(1)}$ & \bigO{n} & \cmark \\
Alg.~\ref{mn_protocol} (PDSMM$_n$) & $n^{1+o(1)}$ & 5 & \cmark \\
Alg.~\ref{alg:wiretap} (Average) & $n^{1+o(1)}$ & 5 & \cmark \\
\hline
\end{tabular}
\endgroup

}
\end{table}

\begin{figure}[ht]
\includegraphics[width=\columnwidth]{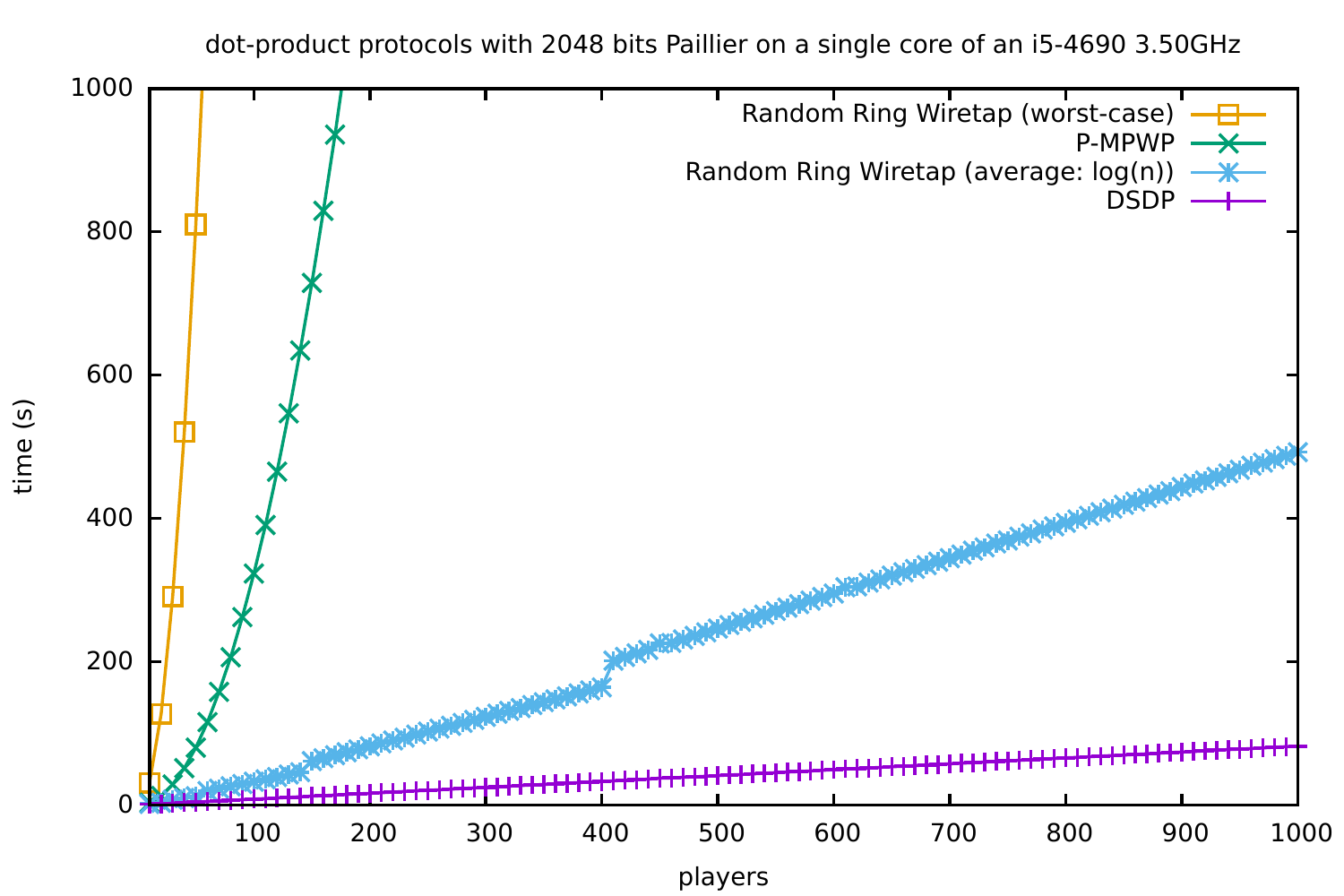}
\caption{Quadratic and linear protocols timings}\label{fig:timcomp}
\end{figure}
On the one hand, we see in Figure~\ref{fig:timcomp} that quadratic
protocols, with homomorphic encryption, are not  usable for a
realistic large group of players (trust aggregation could be used for
instance by certificate authorities, and there are several hundreds of
those in current operating systems or web browsers).
On the other hand, quasi linear time protocols present good
performance, while preserving some reasonable security
properties: the average wiretap curve is on average sufficient to
prevent any attack and still has a quasi linear asymptotic behavior.
The steps in this curve are the rounding of $\log(n)$ to the
next integer and correspond to one more random ring wiretap round.

\section{\uppercase{Conclusion: MPC of Trust}}\label{sec:mpcmonoid}
We now come back to the aggregation of trust.
As shown in Section~\ref{sec:monoid}, the first step is to reduce the
computation to that of dot-products. 
We show how to fully adapt the protocol of
Section~\ref{sec:protocol} to the evaluation of trust values with
parallel and sequential aggregations:
\begin{corollary} The protocol $DSDP$ of Algorithm~\ref{tn_protocol} can be
  applied on trust values, provided that the random values $r_i$ are
  invertible for $\paragg$. 
\end{corollary}
\begin{proof}
\begin{compactitem}
\item $u_i$, $v_i$, $r_i$, $c_i$, $\alpha_i$, $\beta_i$,
  $\Delta_i$, $\gamma$ are now couples;
\item Encryption and decryption ($E(v_i)$, $D(\beta_i)$,
  $E(\Delta_i)$, $E(\gamma)$, etc.) now apply on couples, using the
morphism~$E(\couple{a}{b})=\couple{E(a)}{E(b)}$;
\item $\alpha_i$ is  $E((u_i\seqagg{}v_i)\paragg{}r_i)=\HAdd{\HMul{E(v_i)}{u_i}}{r_i}$, and can
  still be computed by $P_1$, since $c_i=E(v_i)$ and $u_i$ and $r_i$
  are known to him;
\item Similarly,
  $\beta_i=E(\alpha_i\paragg{}\Delta_i)=\HAdd{E(\alpha_i)}{\Delta_i}$.
\item Finally, as \paragg is commutative, $S$ is recovered by adding
  the inverses for $\paragg$ of the $r_i$.
\end{compactitem}
\end{proof}

From~\cite[Definition~11]{jgd:2012:pkitrust},
the $d$-aggregation of trust is a dot-product but slightly modified
to {\em not} include the value $u_1v_1$. Therefore at
line~\ref{sinit}, in the protocol of
Algorithm~\ref{mn_protocol}, it suffices to set $s$ to the neutral
element of \paragg (that is $s\leftarrow\couple{0}{1}$, instead of
$s\leftarrow{}a_{i,j}b_{i,j}$).
There remains to encode trust values that are proportions, in $[0,1]$,
into $\D=\ZN$.
With $n$ participants, we use a fixed precision $2^{-p}$ such that
$2^{n(2p+1)}<N\leq{}2^{n(2(p+1)+1)}$ and round the trust
coefficients to $\lfloor{}x2^p\rfloor\mod N$ from $[0,1] \rightarrow
\D$.
Then the dot-product can be bounded as follows:
\begin{lemma}\label{lem:precision}
If each coefficient of the $u_i$ and $v_i$ are between $0$ and
$2^{p}-1$, then the coefficients of $S=\paragg_{i=1}^n
(u_i\seqagg{}v_i)$ are bounded by $2^{n(2p+1)}$ in absolute value.
\end{lemma}
\begin{proof}
For all $u,v$, the coefficients of $(u\seqagg{}v)$ are between $0$ and
$(2^p-1)(2^p-1)+(2^p-1)(2^p-1)=2^{2p+1}-2^{p+2}+2<2^{2p+1}-1$ for
$p$ a positive integer.
Then, by induction, when aggregating $k$ of those with \paragg, the
absolute values of the 
coefficients remain less than $2^{k(2p+1)}-1$.
\end{proof}

Therefore, with $N$ an $2048$ bits modulus and
$n\leq{}4$ in the $ESDP$ protocols of Algorithm~\ref{mn_protocol},
Lemma~\ref{lem:precision} allows a precision close to $2^{-255}\approx{}10^{-77}$.

In conclusion, we provide an efficient and secure protocol $DSDP_n$ to
securely compute dot products (against semi-honest adversary) in the
MPC model, with unsual data division between $n$ players. It can be
used to perform a private matrix multiplication and also be adapted to
securely compute trust aggregation between players.

\bibliographystyle{plainurl}
\bibliography{smpmatmul}
\end{document}